\documentclass{article}
\usepackage[utf8]{inputenc}
\usepackage{amsmath,amsthm}
\usepackage{amssymb}
\usepackage{mathtools}
\usepackage{mathrsfs}
\usepackage{hyperref}
\usepackage{color}
\hypersetup{colorlinks=true,citecolor=blue,linkcolor=blue,filecolor=blue,urlcolor=blue}
\usepackage{graphicx}
\usepackage{subfig} 
\usepackage{verbatim}
\usepackage{appendix}
\usepackage{cleveref}
\usepackage{mdframed}

\newcommand{\Zp}{\ensuremath{\mathsf{Z}}}

\renewcommand{\sec}[1]{\hyperref[sec:#1]{Section~\ref*{sec:#1}}}

\newcommand{\im}{\mathrm{i}}

\newcommand{\N}{\ensuremath{\mathbb{N}}}

\newcommand{\R}{\ensuremath{\mathbb{R}}}

\def\lrb#1{\left( {#1} \right)}

\newtheorem{theorem}{Theorem}

\newtheorem{corollary}{Corollary}

\newtheorem{proposition}{Proposition}

\crefname{figure}{Fig.}{Figs.}
\Crefname{figure}{Fig.}{Figs.}

\crefname{theorem}{Theorem}{Theorems}
\Crefname{theorem}{Theorem}{Theorems}

\crefname{proposition}{Proposition}{Propositions}
\Crefname{proposition}{Proposition}{Propositions}

\crefname{corollary}{Corollary}{Corollaries}
\Crefname{corollary}{Corollary}{Corollaries}

\DeclarePairedDelimiter\bra{\langle}{\rvert}
\DeclarePairedDelimiter\ket{\lvert}{\rangle}
\DeclarePairedDelimiterX\braket[2]{\langle}{\rangle}{#1 \delimsize\vert #2}
\DeclarePairedDelimiterX\ketbra[2]{| }{|}{#1 \delimsize\rangle\!\delimsize\langle #2}
\DeclarePairedDelimiterX\dotp[2]{\langle}{\rangle}{#1, #2}

\title{A simple lower bound for the complexity of estimating partition functions on a quantum computer}
\author{
    Zherui Chen\thanks{Xingjian College, Tsinghua University. Email: \href{mailto:zr-chen20@mails.tsinghua.edu.cn}{zr-chen20@mails.tsinghua.edu.cn}}
    \and
    Giacomo Nannicini\thanks{Department of Industrial and Systems Engineering, University of Southern California. Email: \href{mailto:g.nannicini@usc.edu}{g.nannicini@usc.edu}}
}
\date{\today} % This will remove the date

\begin{document}
\maketitle
\begin{abstract}

    We study the complexity of estimating the partition function $\Zp(\beta)=\sum_{x\in\chi} e^{-\beta H(x)}$ for a Gibbs distribution characterized by the Hamiltonian $H(x)$. We provide a simple and natural lower bound for quantum algorithms that solve this task by relying on reflections through the coherent encoding of Gibbs states.
    Our primary contribution is a $\varOmega(1/\epsilon)$ lower bound for the number of reflections needed to estimate the partition function with a quantum algorithm. The proof is based on a reduction from the problem of estimating the Hamming weight of an unknown binary string.
    
\end{abstract}

\section{Introduction}\label{sec:introduction}

In the context of statistical physics, a Gibbs distribution is a probability distribution that is used to calculate the statistical properties of a system in thermodynamic equilibrium. The Gibbs distribution is characterized by a Hamiltonian $H(x)$, which is the energy of the system's configuration $x$ and describes the likelihood that the system occupies the configuration. Usually, a temperature parameter is also given, and it affects the probabilities. The partition function associated with the Gibbs distribution is the normalization factor ensuring that the probabilities add up to one. (Formal definitions are given later given in the paper.) The partition function can also encode important physical information of the system, such as the free energy, the average energy and the entropy \cite{reichl2016modern,chowdhury2021computing}.

Besides the applications in physics or statistical mechanics \cite{gardiner2004quantum,trushechkin2022open,albeverio2002euclidean}, Gibbs distributions and the partition functions are also used in mathematics and computer science. Examples of applications of partition functions include the convergence analysis of stochastic gradient descent \cite{shi2020learning,shi2021hyperparameters}, counting proper $k$-colorings of a graph \cite{jerrum1995very, vigoda1999improved,bhandari2020improved}, the ferromagnetic Ising model \cite{mossel2013exact,weinberg2020scaling}, as well as counting matchings and independent sets \cite{harrow2020adaptive,arunachalam2022simpler,cornelissen2023sublinear}. Some of these applications depend on efficiently estimating the value of the partition function at different temperatures. While it is relatively easy to estimate at high temperatures, estimation becomes more challenging at low temperatures, although the low-temperature regime is where the aforementioned applications are found. For instance, the number of proper $k$-colorings of a graph \cite{jerrum1995very,vstefankovivc2009adaptive} and the volume of a convex body \cite{dyer1991computing,chakrabarti2023quantum} are obtained from the partition function at very low temperatures.

The literature contains many algorithms, both classical and quantum, to estimate partition functions. The majority of these approaches employ a strategy that begins in a high-temperature region and progressively lowers the temperature, using a telescoping product of partition functions at neighboring temperatures. Initial studies primarily employed a non-adaptive cooling schedule, implying that the gaps between neighboring temperatures are constant \cite{dyer1991computing,bezakova2008accelerating}; subsequently, adaptive schedules were proposed \cite{vstefankovivc2009adaptive,huber2015approximation,montanaro2015quantum,kolmogorov2018faster,harris2020parameter,harrow2020adaptive,arunachalam2022simpler,cornelissen2023sublinear}.

Most studies focus on improving the upper bound of partition function estimation. The study of lower bounds on the complexity of estimating partition function has not received as much attention. Exceptions are \cite{vstefankovivc2009adaptive,kolmogorov2018faster,childs2022quantum,harris2020parameter}, which, however, do not directly address the question of providing a quantum lower bound for the problem in a general fashion: we discuss the lower bounds of these papers in a subsequent section, after introducing the necessary notation. The  derivation of a general lower bound was also stated as an open question in \cite{arunachalam2022simpler}. Although a quantum lower bound of $\varOmega(1/\epsilon)$ seems natural and almost unavoidable in this setting, to the best of our knowledge a formal proof is not discussed in the literature. It is possible that such a result is part of the folklore: in that case, we welcome readers who are aware of existing proofs to contact us so that we can give proper attribution.

In this note, using a simple argument, we establish a $\varOmega(1/\epsilon)$ lower bound on the number of reflections needed by a quantum algorithm for estimating partition functions at low temperature. Existing quantum algorithms \cite{montanaro2015quantum,harrow2020adaptive,arunachalam2022simpler,cornelissen2023sublinear} match this scaling in $\epsilon$. With a similar argument, we also establish a $\varOmega(1/\epsilon^2)$ query lower bound for classical algorithms for estimating partition functions; this is weaker than the lower bound of $\varOmega(\log |\chi|/\epsilon^2)$ samples from the Hamiltonian recently given in \cite{harris2020parameter}, but we mention it because the proof is considerably simpler. The fastest classical algorithms \cite{kolmogorov2018faster,harris2020parameter} match this scaling in $\epsilon$.

The rest of this paper is organized as follows. In the remainder of this section we introduce our notation, discuss the existing lower bounds, and summarize our results (\cref{thm:main} and \cref{thm:main_classical}). In Section \ref{sec:quantum_lower_bound}, we prove the quantum lower bound (\cref{thm:main}). In Section \ref{sec:classical_lower_bound}, we prove the classical lower bound (\cref{thm:main_classical}). In Section \ref{sec:conclusion}, we conclude our work and discuss some open problems.

\subsection{Notation} Throughout this paper, we use $\cal O$, $o$, $\varOmega$ and $\varTheta$ to denote the traditional computer science notation for the asymptotic behavior of functions. Given an $n$-digit binary string $b \in \{0,1\}^n$, $\|b\|_1=\sum_{i=1}^n b_i$ denotes its Hamming weight. For $n \in \N$, $[n] = \{1,\dots,n\}$.

\subsection{Basic definitions} \label{sec:basic_definitions}

Let $\chi$ be a finite set  and $H:\chi\to\{0,1,\cdots,n\}$ a function called the \emph{Hamiltonian}. (The ground set for the Hamiltonian is often called $\varOmega$ in the literature, but we use $\chi$ because we reserve $\varOmega$ for the asymptotic notation for lower bounds.)  The \emph{Gibbs distribution} is defined as
\begin{align}
    \mu_\beta(x)=\frac{1}{\Zp(\beta)}e^{-\beta H(x)} \qquad \forall x\in\chi,
    \label{eq:gibbs}
\end{align}
where $\beta >0$ is called \emph{inverse temperature} and the normalization factor $\Zp(\beta)$ is the \emph{partition function}, explicitly:
\begin{align*}
    \Zp(\beta)=\sum_{x\in\chi} e^{-\beta H(x)}.
\end{align*}
The goal of \emph{partition function estimation} is to determine a value $\hat{Z}$ that estimates $\Zp(\beta)$ with a small relative error $\epsilon > 0$ at some given $\beta$, i.e., $(1-\epsilon)\Zp(\beta)\leq\hat{Z}\leq (1+\epsilon)\Zp(\beta)$.

The coherent quantum encoding of a Gibbs distribution is given by the following quantum state, also known as \emph{qsample} or \emph{Gibbs state}:
\begin{align}
    \ket{\mu_\beta}=\sum_{x\in \chi}\sqrt{\mu_\beta (x)}\ket{x}.
    \label{eq:mu_beta_state}
\end{align}
A measurement in the computational basis from \eqref{eq:mu_beta_state} yields a sample from the Gibbs distribution \eqref{eq:gibbs}. These qsamples are used in existing quantum algorithm to estimate $\Zp(\beta)$, see, e.g., \cite{montanaro2015quantum,harrow2020adaptive,arunachalam2022simpler,cornelissen2023sublinear}.

The computation of $\Zp(\beta)$ is relatively straightforward in the high-temperature region, where $\beta$ is small. For $\beta = 0$, the calculation is usually trivial. Conversely, in the low-temperature region, where $\beta$ is large, the calculation is difficult in general. We focus on the low-temperature region, and discuss the complexity of computing $\Zp(+\infty)$.

\subsection{Existing lower bounds}
For classical algorithms, the first lower bound on partition function
estimation is due to
\cite{vstefankovivc2009adaptive}. \cite{vstefankovivc2009adaptive}
gives a lower bound on the number of temperatures of a non-adaptive
cooling schedule, and therefore this lower bound only applies to the
class of simulated annealing algorithms with non-adaptive cooling
schedules. \cite{kolmogorov2018faster} establishes a lower bound of
$\varOmega\left(\frac{\log{|\chi|}}{\epsilon^2}\right)$ for the
estimation of partition function, in a computational model where we
receive samples from $H(x)$ but we are now allowed to know the
corresponding $x$, or query $H(x)$ for specific $x$. \cite{harris2020parameter} extends this to a lower bound in the black-box sampling model, showing optimality of existing classical algorithms. (Our classical lower bound
is not as strong, but it has a very simple proof.)

Regarding quantum algorithms, the only lower bound that we are aware
of is \cite{childs2022quantum}, which gives a lower bound of
$\varOmega(\epsilon^{-\frac{1}{1+4/k}})$ for the quantum query
complexity of estimating $\int_{\R^k} e^{-f(x)} \textrm{d}x$, under
the assumption that $f$ is $1.5$-smooth and $0.5$-strongly
convex. This lower bound is different from ours in that we consider a
general Hamiltonian defined over a discrete set, whereas
\cite{childs2022quantum} considers $f$ defined over real space and
strongly convex. Note that a classical lower bound of
$\varOmega(k/\epsilon^2)$ for the same setting as \cite{childs2022quantum} is proven in \cite{ge2020estimating}.

\subsection{Contributions}

Our main contribution (\cref{thm:main}) is a quantum query complexity lower bound for estimating partition functions with a simple proof. To state the lower bound we need to identify a meaningful ``unit of work'' for quantum algorithms to estimate $\Zp(+\infty)$. The basic steps that \cite{harrow2020adaptive,arunachalam2022simpler,cornelissen2023sublinear} have in common is the reflection through the Gibbs states \eqref{eq:mu_beta_state} at different inverse temperatures $\beta$, therefore our lower bound is stated in terms of the number of reflection through these states. Note that the quantum algorithms described in the above papers access the Hamiltonian only through these reflections.
\begin{theorem}[Informal]\label{thm:main}
    Any quantum algorithm that computes $\hat{Z}$ satisfying $|\hat{Z}-\Zp(+\infty)|\leq \epsilon \cdot\Zp(+\infty)$ for a given Hamiltonian $H$, where $H$ is accessed via the reflection through Gibbs states \eqref{eq:mu_beta_state} at different temperatures, takes $\varOmega({1/\epsilon})$ such reflections.
\end{theorem}
We also give a query lower bound of estimating partition function using classical algorithms. 
\begin{theorem}[Informal]\label{thm:main_classical}
    Any classical algorithm that computes $\hat{Z}$ satisfying $|\hat{Z}-\Zp(+\infty)|\leq \epsilon \cdot\Zp(+\infty)$ for a given Hamiltonian $H$, takes $\varOmega({1/\epsilon^2})$ queries to $H$.
\end{theorem}
The classical lower bound is weaker than the lower bound of $\varOmega({\log |\chi|/\epsilon^2})$ given in \cite{harris2020parameter}, but the proof is simple. The results of \cite{harris2020parameter} are very recent and we were not aware of them when working on our result, hence we did not try to quantize them. It is possible that their construction yields a stronger quantum lower bound than Thm.~\ref{thm:main}.

\section{Quantum lower bound} \label{sec:quantum_lower_bound}
In this section we establish the quantum lower bound on the number of reflection operators to estimate the partition function. In Section \ref{sec:qu_preliminaries} we first introduce the quantum query complexity of the Hamming weight problem, and discuss fixed-point quantum search. Then, we prove the lower bound (\cref{thm:main}) by constructing a finite set $\chi$ and a Hamiltonian that encode the Hamming weight problem, see Section \ref{sec:construction}.

\subsection{Preliminaries}\label{sec:qu_preliminaries}
Our starting point is a well-known quantum query bound for the problem of estimating the Hamming weight of an unknown binary string. This is discussed in \cref{prop:hamming_weight} and \cref{cor:hamming_weight}, and used in the proof of \cref{thm:main}.

\begin{proposition}[Quantum query tight bound for the Hamming weight problem \cite{nayak1999quantum} \cite{childs2022quantum}]\label{prop:hamming_weight}
    Let $b=\{b_1,b_2,\cdots,b_N\}\in\{0,1\}^N$ be an $N$-digit binary string, and let $\ell$, $\ell^{\prime}$  be integers such that  $0 \leq \ell<\ell^{\prime} \leq N$. Define the partial Boolean function  $f_{\ell, \ell^{\prime}}$  on a subset of $\{0,1\}^{N}$ as follows:
    $$
        f_{\ell, \ell^{\prime}}(b)=\left\{\begin{array}{ll}
            0 & \text { if }\|b\|_{1}=\ell            \\
            1 & \text { if }\|b\|_{1}=\ell^{\prime} .
        \end{array}\right.
    $$
    Let $m \in\left\{\ell, \ell^{\prime}\right\}$  be such that  $\left|\frac{N}{2}-m\right|$  is maximized, and let  $\Delta=\ell^{\prime}-\ell$. Then, given the quantum query oracle
    \begin{align}
        O_{b}\ket{i}\ket{y}=\ket{i, y \oplus b_{i}}\quad \forall i \in[N], y \in\{0,1\},
        \label{eq:oracle_O_b}
    \end{align}
    the number of queries to $O_b$ for computing the function  $f_{\ell, \ell^{\prime}}$  is $\varTheta(\sqrt{N / \Delta}+\sqrt{m(N-m)} / \Delta)$ .
\end{proposition}
\begin{corollary} \label{cor:hamming_weight}
    Let $\delta >0$ such that $2N\delta\geq 1$. Assume that we are given quantum query access $O_{b}$ to $b \in \{0,1\}^N$, where $\|b\|_1=\ell=N(\frac{1}{2}+\delta)$ or $\|b\|_1=\ell'=N(\frac{1}{2}-\delta)$. Then the quantum query complexity of computing the function $f_{\ell, \ell^{\prime}}$ (i.e., distinguishing the two cases for $\|b\|_1$) is $\varTheta({1/\delta})$.
\end{corollary}
\begin{proof}
    Applying \cref{prop:hamming_weight}, $m \in\{N(\frac{1}{2}+\delta) ,N(\frac{1}{2}-\delta)\}$ and $\Delta =2 N\delta$. Thus, the quantum query complexity  is  $\varTheta(\sqrt{N / \Delta}+\sqrt{m(N-m)} / \Delta)=\varTheta(\sqrt{\frac{N}{N\delta}}+\sqrt{(\frac{N}{2}+\delta N)(\frac{N}{2}-N\delta)} / (2N\delta))= \varTheta({1/\delta})$.
\end{proof}

We also need a result about fixed-point quantum search, a generalization of standard Grover search \cite{grover1996fast} introduced in \cite{yoder2014fixed}. We summarize the main result of \cite{yoder2014fixed} in \cref{prop:fixed-point}.
\begin{proposition} [Fixed-point quantum search \cite{yoder2014fixed}]\label{prop:fixed-point}
    Let $A$ be a unitary such that $\ket{s}=A\ket{0}^{\otimes n}$. Let $\ket{T}$ be a target state such that $\braket{T}{s}=\sqrt{\lambda} e^{\im \xi}$ for some $\xi$ and $\lambda > 0$. Let $U$ be a unitary such that $U \ket{T} \ket{y} = \ket{T}\ket{y \oplus 1}$ and $U \ket{T^{\perp}}\ket{y} = \ket{T^{\perp}}\ket{y}$ for $\braket{T}{T^{\perp}} = 0$. Let $\eta \in (0,1]$. There exists a quantum circuit $F$, consisting of
    \begin{equation*}
        L={\cal O}\left(\frac{\log (2/\eta )}{\sqrt{\lambda}}\right)
    \end{equation*}
    calls to $A,A^\dagger,U$ and efficiently implementable $n$-qubit gates, such that $|\bra{T}F\ket{s}|^2 \ge 1- \eta^2$.
\end{proposition}
Note that the above proposition implies that we construct the target state with probability at least $1 - \eta^2$.

\subsection{Construction of the Hamiltonian}\label{sec:construction}
We prove \cref{thm:main} by constructing a specific finite set $\chi$ and a Hamiltonian in such a way that we can reduce the problem of computing $f_{\ell,\ell'}$ to the problem of partition function estimation. The reduction gives us a lower bound on the number of queries to the Hamiltonian $H$. We then use fixed-point quantum search to convert this lower bound to a lower bound on the number of reflections through Gibbs states of the form $\ket{\mu_\beta}$. We state a more accurate version of \cref{thm:main} here:
\begin{theorem}[Quantum lower bound for estimating partition function]\label{thm:main_2}
    Let $H:\chi\to \{0,1,\cdots,n\}$ be a Hamiltonian and let $\Zp(\beta)$ be the corresponding partition function. Let ${\cal A}$ be the class of quantum algorithms that estimate $\Zp$ by reflecting through a Gibbs state at arbitrary inverse temperatures $\beta$ at most ${\cal O}(\log^2 |\chi|/\epsilon)$ times, and that query $H$ only to perform such reflections. Any quantum algorithm from the class ${\cal A}$ that estimates $\hat{Z}$ satisfying $|\hat{Z}-\Zp(+\infty)|\leq \epsilon \cdot\Zp(+\infty)$ requires $\varOmega\lrb{{1}/{\epsilon}}$ reflections.
\end{theorem}
The restriction that reflections are invoked at most ${\cal O}(\log^2 |\chi|/\epsilon)$ times is technical: known quantum algorithms for partition functions, e.g., \cite{montanaro2015quantum,harrow2020adaptive,arunachalam2022simpler,cornelissen2023sublinear}, achieve this scaling (possibly up to polylogarithmic factors, which we neglect for simplicity) or better.
\begin{proof}
    Let $b=\{b_1,b_2,\cdots,b_N\}\in\{0,1\}^N$. Define $\chi=\{1,2,\cdots, N\}$, and $H_b(x)$ as:
    \begin{equation*}
        H_b(x)= 1-b_x, \qquad \text{ for } x \in \chi.
    \end{equation*}
    Then $H_b(x) = 0$ if and only if $b_x = 1$, and therefore
    \begin{align}
        \Zp(+\infty)=\lim _{\beta\to+\infty}\sum _{x\in \chi}e^{-\beta H_b(x)}=\|b\|_1.
        \label{eq:Zp_inf_b}
    \end{align}
    Let $O_{H}$ be a quantum oracle to evaluate $H_b(x)$, defined as:
    \begin{align}
        O_{H}\ket{x}\ket{y}=\ket{x}\ket{y\oplus {H_b}(x)}, \text{ for }\forall x\in \chi, y\in \{0,1\}.
        \label{eq:oracle_O_H}
    \end{align}
    Note that we can simulate one query to $O_H$ with one query to $O_b$. (We drop the double-subscript $b$ for $O_{H_b}$ or whenever the particular $b$ under consideration is clear from the context.)

    Let $b^{(1)}$ and $b^{(2)}$ be two binary strings satisfying $\|b^{(1)}\|_1=N(\frac{1}{2}-\delta)$ and $\|b^{(2)}\|_1=N(\frac{1}{2}+\delta)$. Given query access to $b^{(1)}$ or $b^{(2)}$, the partition function $\Zp(\infty)$ for Hamiltonians $H_{b^{(1)}}, H_{b^{(2)}}$ allows us to compute the Hamming weight of the corresponding strings, and therefore, estimating $\Zp(\infty)$ with sufficient precision allows us to distinguish whether we were given $b^{(1)}$ or $b^{(2)}$. If we want to distinguish these strings, it is sufficient to estimate $\Zp(\infty)$ with precision $\epsilon = \delta$, because:
    \begin{align*}
        \frac{1}{2}\frac{\|b^{(2)}\|_1-\|b^{(1)}\|_1}{\Zp(+\infty)} \ge \frac{1}{2}\frac{2N\delta}{N\left(\frac{1}{2}+\delta\right)} \ge \delta 
    \end{align*}
    for any $\delta < 1/2$.
    We fix $\epsilon = \delta$ for the rest of this proof. By \cref{cor:hamming_weight}, distinguishing $b^{(1)}$ and $b^{(2)}$ requires at least $\varOmega(1/\delta)=\varOmega(1/\epsilon)$ queries to $O_b$ \eqref{eq:oracle_O_b}. Since one query to $O_H$ takes only one query to $O_b$, estimating $\Zp(\infty)$ with precision $\epsilon$ also takes $\varOmega(1/\epsilon)$ queries to $O_H$.

    We now convert the lower bound for the number of queries to $O_H$ to a lower bound for the number of reflections through Gibbs states $\ket{\mu_\beta}$. We do so by giving an upper bound on the number of queries to $O_H$ that allow us to construct $\ket{\mu_\beta}$ from scratch, and it is well known that if we can construct a state from the all-zero state, we can also reflect through the desired state by using the reflection through the all-zero state (see, e.g., Grover search \cite{grover1996fast} or amplitude amplification \cite{brassard2002quantum}). We use  fixed-point quantum search to construct $\ket{\mu_\beta}$ using $O_H$, relying on \cref{prop:fixed-point}.

    By definition:
    \begin{align}
        \ket{\mu_{\beta}}=\frac{1}{\sqrt{\Zp({\beta})}} \left (\overbrace{ \sum_{x : b_x=1} \ket{x} }^{\|b\|_1 \text{ states}} ~+~\overbrace{\sum_{x : b_x=0} \sqrt{e^{-\beta}} \ket{x}}^{ N - \|b\|_1 \text{ states}}  \right),
        \label{eq:mu_beta}
    \end{align}
    and $\Zp({\beta})=\|b\|_1+(N-\|b\|_1)e^{-\beta}$ is the partition function at inverse temperature $\beta$.
    Setting the target state $\ket{T}$ of \cref{prop:fixed-point} to be
    $$\ket{T}=\frac{1}{\sqrt{\|b\|_1}}\sum_{x : b_x=1} \ket{x} = \ket{\mu_{+\infty}},$$
    it is immediate to observe that the marking oracle $U$ of \cref{prop:fixed-point} can be implemented with a single call to $O_H$ \eqref{eq:oracle_O_H}, as they have the same action of flipping an ancilla qubit when applied onto $\ket{T}$.

    We first give an upper bound for the number of queries to $O_H$ to construct $\ket{\mu_{+\infty}}$ using \cref{prop:fixed-point}, and subsequently, we prove that with the same cost we can construct $\ket{\mu _{\beta}}$ for any other $\beta$. For now we assume that $\|b\|_1$ is know, and later we show that an unknown value of $\|b\|_1$ has negligible impact. Our initial state $\ket{s}$ is the uniform superposition of all states $\ket{x}$:
    \begin{align*}
        \ket{s}=\frac{1}{\sqrt{N}}\sum_{x=1}^{N}\ket{x}.
    \end{align*}
    The overlap between $\ket{s}$ and $\ket{T}$ is:
    \begin{align*}
        \braket{T}{s} =\sqrt{\frac{\|b\|_1}{N}} =\sqrt{\lambda}.
    \end{align*}
    Choosing $\eta_{+\infty}$ to obtain a constant success probability $1-\eta_{+\infty}^2$ for \cref{prop:fixed-point}, say, 99.9\%, by \cref{prop:fixed-point} the number of queries to $O_H$ to construct $\ket{T}$ is
    \begin{align*}
        L_{+\infty}={\cal O}\left( \log\lrb{\frac{2}{\eta_{+\infty}}} \sqrt{\frac{N}{\|b\|_1}}\right)={\cal O}\left( \sqrt{\frac{N}{\|b\|_1}}\right).
    \end{align*}
    Constructing $\ket{\mu_\beta}$ for any other $\beta < \infty$ does not increase the query cost. Rewrite \eqref{eq:mu_beta} as:
    \begin{align*}
        \ket{\mu_\beta}=\frac{1}{\sqrt{\Zp(\beta)}}\lrb{\sqrt{\|b\|_1}\ket{T}+\sqrt{\lrb{N-\|b\|_1}e^{-\beta}}\ket{T^{\perp}}}.
    \end{align*}
    In the context of quantum singular value transformation / quantum signal processing, we can use the same odd polynomial as for fixed-point quantum search, but rather than aiming for a final amplitude value of $\ket{T}$ in the interval $[1-\eta, 1]$, we scale it down to an interval of size $O(\eta)$ around the desired target $\sqrt{\|b\|_1}/\sqrt{\Zp(\beta)}$, see \cite{low2017hamiltonian,gilyen2019quantum}. 
    
    One potential issue still remains regarding the construction of a reflection via state preparation: we do not know $\|b\|_1$, therefore the state preparation above cannot be exact. We show that the overlap between the ideal state (assuming knowledge of $\|b\|_1$) and the imperfect state (assuming no knowledge of $\|b\|_1$) is $1-{\cal O}(\delta)$, and this does not significantly affect the final result.
    
    Suppose we have $\|b^{(1)}\|_1=N\lrb{\frac{1}{2}-\delta }$, but we implement the reflection using the above state-preparation algorithm with the wrong value $\|b^{(2)}\|_1=N\lrb{\frac{1}{2}+\delta }$. We can still construct $\ket{\mu_{\beta}}$ for any $\beta$ with probability at least $1-4\delta^2$, as demonstrated in the following calculation.
    The ideal state is
    $$\ket{\mu_{\beta}}=\frac{1}{\sqrt{N\lrb{\frac{1}{2}-\delta }+N\lrb{\frac{1}{2}+\delta }e^{-\beta}}}\lrb{\sqrt{N\lrb{\frac{1}{2}-\delta }}\ket{T}+\sqrt{{N\lrb{\frac{1}{2}+\delta }}e^{-\beta}}\ket{T^{\perp}}},$$
    and the imperfect state is
    $$\ket{\mu_{\beta}^{\text{w}}}=\frac{1}{\sqrt{N\lrb{\frac{1}{2}+\delta }+N\lrb{\frac{1}{2}-\delta }e^{-\beta}}}\lrb{\sqrt{N\lrb{\frac{1}{2}+\delta }}\ket{T}+\sqrt{{N\lrb{\frac{1}{2}-\delta }}e^{-\beta}}\ket{T^{\perp}}},$$
    where $\ket{T}=\frac{1}{\sqrt{N\lrb{\frac{1}{2}-\delta }}}\sum_{x : b^{(1)}_x=1} \ket{x}$. The overlap between the two states is
    \begin{align}
        \left|\braket{\mu_{\beta}}{\mu_{\beta}^{\text{w}}}\right| & =\frac{\sqrt{N^2\lrb{\frac{1}{4}-\delta^2}} +  \sqrt{N^2 (\frac{1}{4}-\delta^2)e^{-2\beta}}}{\sqrt{N\lrb{\frac{1}{2}-\delta }+N\lrb{\frac{1}{2}+\delta }e^{-\beta}}\sqrt{N\lrb{\frac{1}{2}+\delta }+N\lrb{\frac{1}{2}-\delta }e^{-\beta}}}\nonumber\\
        & =\frac{\sqrt{{\frac{1}{4}-\delta^2}}  \lrb{1+e^{-\beta}}}{\sqrt{\frac{1}{4} \lrb{1+e^{-\beta}}^2- \delta^2 \lrb{1-e^{-\beta}}^2}}\nonumber\\
        & = \frac{\sqrt{{\frac{1}{4}-\delta^2} }}{\sqrt{\frac{1}{4}-\delta^2 \lrb{\frac{1-e^{-\beta}}{1+e^{-\beta}}}^2}}.
        \label{eq:overlap_1}
    \end{align}
    This is an increasing function of $\beta$. For $\beta =0$ we have $\left|\braket{\mu_{0}}{\mu_{0}^{\text{w}}}\right| =\sqrt{1-4\delta^2}$. Hence, we can construct $\ket{\mu_{\beta}}$ for any $\beta$ with probability at least $1-4\delta^2$. (In fact, for $\beta \to +\infty$, we have $\left|\braket{\mu_{+\infty}}{\mu_{+\infty}^{\text{w}}}\right| =1$.) Suppose now we have $\|b^{(2)}\|_1=N\lrb{\frac{1}{2}+\delta }$, but we implement the state preparation with the wrong value $\|b^{(1)}\|_1=N\lrb{\frac{1}{2}-\delta }$. \eqref{eq:overlap_1} still holds, and we can construct $\ket{\mu_{\beta}}$ for any $\beta$ with probability at least $1-4\delta^2$. This shows that the construction of $\ket{\mu_\beta}$, and the corresponding reflection, succeeds with probability at least $1-{\cal O}(\delta^2)$ even without knowing $\|b\|_1$.

    Thus, for every given $\delta$, we construct an instance of a partition function problem whose solution computes the function $f_{\ell,\ell'}$ in Cor.~\ref{cor:hamming_weight}. To do so, we need to to choose $N \ge 1/(2\delta)$. Let $m$ be the number of queries to $O_H$ performed by the partition function estimation algorithm, and let $1-p$ be the probability of success of the algorithm for some $p < 1/4$. Using inexact reflections can lower the probability of success. Recall that we assume $m = {\cal O}(\log^2 |\chi|/\epsilon) = {\cal O}(\log^2 N/\delta) \le c \log^2 N /\delta$. Because $p$ and $c$ are constants and $\delta \to 0$, for any given $\delta$ below a certain threshold we can always choose $N$ satisfying:
    \begin{equation*}
        \frac{1}{2N} \le \delta \le \frac{p}{c \log^2 N}.
    \end{equation*}    
    Then we have
    \begin{equation*}
        \delta = \sqrt{\delta} \sqrt{\delta} \le \sqrt{\delta} \sqrt{\frac{p}{c \log^2 N}} \le 
        \sqrt{\delta} \sqrt{\frac{p}{\delta m}} = \sqrt{\frac{p}{m}},
    \end{equation*}
    ensuring that the total failure probability due to the inexact reflections is at most $\delta^2 m = p$. Hence, the partition function estimation algorithm will still be successful with probability at least $1-2p > 1/2$, therefore we can ignore the issue of inexact reflections due to not knowing $\|b\|_1$ in advance.

    Putting everything together, the final lower bound in terms of the number of reflections is
    \begin{align*}
          & \frac{\text{lower bound on the number of $O_{H}$ evaluations to estimate partition function}}{\text{upper bound on the number of $O_{H}$ evaluations to perform an arbitrary reflection}}                                           % \label{eq:words-1}  \nonumber          
          \\
        = & \frac{\varOmega(1/\epsilon)}{{\cal O}\lrb{\sqrt{\frac{N}{\|b\|_1}}}}= \varOmega\lrb{\frac{1}{\epsilon}\cdot\sqrt{\frac{\|b\|_1}{N}}}=\varOmega\lrb{\frac{1}{\epsilon}\cdot\sqrt{\frac{N/2 - \delta N}{N}}}={\varOmega}\lrb{\frac{1}{\epsilon} - \sqrt{\frac{1}{\epsilon}}} = {\varOmega}\lrb{\frac{1}{\epsilon}}. \qedhere
    \end{align*}
\end{proof}

\section{Classical lower bound} \label{sec:classical_lower_bound}
Here we prove a classical Hamiltonian query lower bound based on the construction in Section \ref{sec:construction}.

\subsection{Preliminaries}
Our starting point is a classical query bound.
\begin{proposition}[Biased coin model, Claim D.1 of \cite{ge2020estimating}]\label{prop:biased_coin}
    Suppose we are given independent samples of a random variable $X$, where $X$ is drawn from a Bernoulli distribution with probability either $p=1/2+\delta$ or $p=1/2-\delta$. Then, any algorithm that looks at $o(1/\delta^2)$ samples of $X$ cannot decide, with probability better than $1/2+c$ for any constant $c>0$, which distribution $X$ is drawn from.
\end{proposition}
We also use the Chernoff bound, a version of which is stated next.
\begin{proposition}[Chernoff bound, Theorem 2 of \cite{Lecture3}]\label{prop:chernoff}
    Let $X_1,X_2,\cdots,X_N$ be independent random variables, where $X_i\in\{0,1\}$ and $\Pr[X_i=1]=p_i$ with $0<p_i<1$. Denote $\mu=\sum_{i=1}^N\mathbb{E}[X_i]=\sum_{i=1}^N p_i$. Then for any $0<t<1$, we have
    \begin{align}
        \Pr[\sum_{i=1}^N X_i \leq (1-t)\mu]\leq e^{-\frac{t^2\mu}{2}},\label{eq:chernoff_1} \\
        \Pr[\sum_{i=1}^N X_i \geq (1+t)\mu]\leq e^{-\frac{t^2\mu}{2+t}}. \nonumber
    \end{align}
\end{proposition}

\subsection{Proof}
We prove \cref{thm:main_classical} by using the same construction as in Section \ref{sec:construction}. We restate \cref{thm:main_classical} here:
\begin{theorem}[Classical lower bound for estimating partition function]
    Let $H:\chi\to \{0,1,\cdots,n\}$ be a Hamiltonian and let $\Zp(\beta)$ be its partition function. Let ${\cal A}^{cl}$ be the class of classical algorithms that estimate $\Zp$ by querying the Hamiltonian $H(x)$ at some value $x$. Any classical algorithm from the class ${\cal A}^{cl}$ that estimates $\hat{Z}$ satisfying $|\hat{Z}-\Zp(+\infty)|\leq \epsilon \cdot\Zp(+\infty)$ requires
    $$\varOmega({1}/{\epsilon^2})$$ such
    queries to $H(x)$.
\end{theorem}
\begin{proof}
    Let $\chi$ and $H_b(x)$ be as in the proof of \cref{thm:main_2}. Thus, $\Zp(+\infty)=\|b\|_1$ still holds, see \eqref{eq:Zp_inf_b}. Assume that the digits of the $N$-digit binary string $b$ is drawn from a Bernoulli distribution with probability of $1$ equal to $p=1/2+\delta$ or $p=1/2-\delta$. Note that a query to look at one (new) digit of $b$ is equivalent to looking at one sample from the Bernoulli distribution --- we can assume that we do not query the same digit twice, because we could simply record its value from the previous query.

    We pick $N\geq 20/\delta^2$ so that the following two conditions are satisfied: when $p=1/2+\delta$, we have
    \begin{align}
        \Pr[\|b\|_1\geq N(1/2+\delta/2-\delta^2)]\geq 0.99;
        \label{eq:condition_1}
    \end{align}
    when $p=1/2-\delta$, we have
    \begin{align*}
        \Pr[\|b\|_1\leq N(1/2-\delta/2-\delta^2)]\geq 0.99.
       % \label{eq:condition_2}
    \end{align*}
    We prove \eqref{eq:condition_1}. Eq. \eqref{eq:condition_1} is equivalent to $\Pr[\|b\|_1\leq N(1/2+\delta/2-\delta^2)]\leq 0.01$. In \cref{prop:chernoff}, we set $t=\delta$, which implies
    \begin{align*}
        (1-t)\mu=(1-\delta)N(1/2+\delta)=N(1/2+\delta/2-\delta^2).
    \end{align*}
    Using $N\geq 20/\delta^2$, the right-hand side of \eqref{eq:chernoff_1} becomes
    \begin{align*}
        e^{-\frac{t^2\mu}{2}}=e^{-\frac{\delta^2}{2}N(\frac{1}{2}+\delta)} \le e^{-5}\le 0.01
    \end{align*}
    Thus, \eqref{eq:condition_1} holds. The proof of second condition is analogous.

    Because $$\Zp(+\infty)=\|b\|_1,$$ with probability at least 0.99, we have $\Zp(+\infty)\ge N(1/2+\delta/2-\delta^2):=\tau _1$ when $p=1/2+\delta$, and $\Zp(+\infty)\le N(1/2-\delta/2-\delta^2):=\tau _2$ when $p=1/2-\delta$. Therefore, any algorithm that can estimate $\Zp(+\infty)$ with relative accuracy better than $\epsilon=\frac{\tau_1-\tau_2}{\Zp(+\infty)} \le \frac{N\delta}{N(1/2+\delta/2-\delta^2)} = {\cal O}(\delta)$ (with probability greater than $\frac{1}{2}+c$, where $c>0$ is any constant) would be able to distinguish the two cases for $p$. By \cref{prop:biased_coin}, to solve this task we must look at $\varOmega\left(\frac{1}{\delta^2}\right)=\varOmega\left(\frac{1}{\epsilon^2}\right)$ samples from the Bernoulli distributions, and therefore we need that many queries to the Hamiltonian $H$.
\end{proof}
As previously remarked, \cite{harris2020parameter} gives a stronger sampling lower bound than the one we provide above.

\section{Conclusions}\label{sec:conclusion}
In this paper we prove a quantum lower bound on the number of reflections for estimating partition function at low temperatures. This implies the optimality in the $\epsilon$ parameter of recent quantum algorithms for this problem \cite{montanaro2015quantum,harrow2020adaptive,arunachalam2022simpler,cornelissen2023sublinear}. Our classical lower bound is weaker than the one in \cite{kolmogorov2018faster}, but it uses a more traditional computational model and has a simple proof; \cite{harris2020parameter} supersedes \cite{kolmogorov2018faster} and shows a stronger lower bound in the standard black-box sampling computational model.

The unresolved question is whether we can achieve a lower bound stronger than $\varOmega({1}/{\epsilon})$ for quantum algorithms, in particular incorporating $|\chi|$ into the lower bound. We attempted several constructions for $H$ and $\chi$, and the one discussed here is the most robust among those that we examined. It is possible that a different approach -- one that does not uses a reduction from the problem of determining the Hamming weight of an unknown binary string -- is necessary; in particular, a reduction from a problem with multiple instance parameters involved in the lower bound might be more successful in incorporating $|\chi|$. It is possible that the approach of \cite{harris2020parameter} leads to a stronger quantum lower bound: we did not know their result (which extends \cite{kolmogorov2018faster}) until posting a first version of this manuscript on arXiv, hence we did not make any attempt at quantizing it yet.

\section*{Acknowledgements}
We are grateful to D.~Harris for pointing out \cite{harris2020parameter} to us. G.~Nannicini is partially supported by ONR award \# N000142312585.

%%%%%%%%%%%%%%%%%%%%%%%%%%%%%%%%%%%%%%%%%%%%

\bibliographystyle{alpha}
\bibliography{study-plan-ref}% Produces the bibliography via BibTeX.

%%%%%%%%%%%%%%%%%%%%%%%%%%%%%%%%%%%%%%%%%%%%

\end{document}